\documentclass[a4paper,11pt]{article}
\pdfoutput=1 

\usepackage{jheppub} 

\usepackage[T1]{fontenc} 
\usepackage{amsfonts,amsmath,amsthm,upquote,amssymb,tikz,mathtools}
\usepackage{epstopdf}
\usepackage{pst-node,pst-plot}
\usetikzlibrary{positioning}
\usetikzlibrary{shapes.geometric,calc,through}
\usepackage{color}
\usepackage{mathrsfs}
\allowdisplaybreaks
\newtheorem{thm}{Theorem}[section]

\newtheorem{lem}[thm]{Lemma}

\newcommand{\RR}{{\mathbb R}}

\author[a,b]{Ning Bao,}
\author[b]{ChunJun Cao,}
\author[c]{Michael Walter,}
\author[a,b]{Zitao Wang}

\affiliation[a]{Institute for Quantum Information and Matter, California Institute of Technology, Pasadena, CA 91125, USA}
\affiliation[b]{Walter Burke Institute for Theoretical Physics, California Institute of Technology, Pasadena, CA 91125, USA}
\affiliation[c]{Stanford Institute for Theoretical Physics, Stanford University, Stanford, CA 94305, USA}

\emailAdd{ningbao@theory.caltech.edu}
\emailAdd{cjcao@caltech.edu}
\emailAdd{michael.walter@stanford.edu}
\emailAdd{zwang@caltech.edu}

\abstract{We extend our studies of holographic
entropy inequalities to gapped phases of matter.
For any number of regions, we determine the linear entropy inequalities satisfied by systems in which the entanglement entropy satisfies an exact area law.
In particular, we find that all holographic entropy inequalities are valid in such systems.
In gapped systems with topological order, the ``cyclic inequalities'' derived recently for the holographic entanglement entropy generalize the Kitaev-Preskill formula for the topological entanglement entropy.
Finally, we propose a candidate linear inequality for general 4-party quantum states.}

\begin{document}

\title{Holographic entropy inequalities and gapped phases of matter}

\maketitle
\flushbottom

\section{Introduction}
\label{sec:intro}
 In recent years, the study of quantum entanglement and quantum information in general has produced a myriad of applications in high energy physics and condensed matter physics. A key tool for the quantification of entanglement is, in particular, the entanglement entropy. For general quantum systems, the von Neumann entanglement entropies of subsystems are known to obey subadditivity, the Araki-Lieb inequalities, weak monotonicity, and  strong subadditivity, respectively:

\medskip

\noindent
Subadditivity:
\begin{equation}
S(A)+S(B)-S(AB)\geq  0.
\label{eqn:SA}
\end{equation}
Araki-Lieb:
\begin{equation}
S(C)+S(ABC)-S(AB)\geq 0.
\label{eqn:AL}
\end{equation}
Weak Monotonicity:
\begin{equation}
S(AB)+S(AC)-S(B)-S(C)\geq 0.
\label{eqn:WM}
\end{equation}
Strong Subadditivity:
\begin{equation}
S(AB)+S(AC)-S(A)-S(ABC)\geq 0.
\label{eqn:SSA}
\end{equation}

Such inequalities are important, as they constrain the phase space of entanglement in quantum systems and can in turn be translated into other physical quantities. In particular, in condensed matter physics there exists a conjectured relationship \cite{AreaLawMInfo,Hastings} between the existence of a gap in a system and whether or not the entanglement entropy in that system obeys an area law \cite{AreaLawReview}. It therefore seems a fruitful direction to explore and better characterize the properties of the entanglement entropy in quantum-mechanical systems.

It should be noted, however, that entropy inequalites for general quantum states are relatively rare; indeed, \eqref{eqn:SA}--\eqref{eqn:SSA} are the only unconditional entropy inequalities known to date.
Luckily, there exist classes of quantum systems for which the entanglement entropy is easier to characterize.

\subsection{Entropy inequalities from holography}
In holography, it has been shown that entanglement entropies of regions on the boundary are equal to the areas of the minimal surfaces subtending the boundary region, or in terms of the celebrated Ryu-Takayanagi formula \cite{RT1,RT2}:
\begin{equation}
S(A)=\frac{\text{Area}}{4G}.
\end{equation}

The Ryu-Takayanagi formula gives us a powerful new tool for computing entanglement entropies in regimes where such calculations are usually intractible. In higher dimensions, for example, it turns what would be a difficult (if not impossible) conformal field theory calculation into a straightforward minimization of area in a classical metric.

Interestingly, these holographic entenglement entropies obey a larger set of inequalities than those obeyed by the generic quantum mechanical systems. In \cite{HHM}, it was discovered that, indeed, there is a new entanglement entropy inequality which is true for all systems with semi-classical holographic duals, i.e., the conditional mutual information is monogamous, or
\begin{equation}
S(AB)+S(AC)+S(BC)\geq S(A)+S(B)+S(C)+S(ABC).
\end{equation}

This was done using a method known as inclusion/exclusion, in which minimal surfaces corresponding to the (positively signed) entropic terms on the left hand side are repartitioned into non-minimal surfaces corresponding to the terms on the right hand side. As non-minimal surfaces have more area than minimal surfaces, if such a partitioning can be done, then the inequality is true. A more detailed description of the methodology is available in \cite{HHM}. Recently, this has also been generalized to higher numbers of regions in \cite{BNOSSW} by converting the geometric procedure described above to a combinatoric set of contraction mappings from points on a hypercube, which corresponds to the left hand side entanglement entropies, to points on another hypercube, which corresponds to the right hand side entropies. This new method has yielded a new, infinite family of inequalities that has been proven for holographic systems. These so-called ``cyclic'' inequalities for $n=2k+1$ subsystems take the form
\begin{align}
\text{CYC}&=\sum_{l=1}^{k-1}I(A_1...A_l:A_{k+l+1}:A_{k+l+2}...A_{2k+1})
-\sum_{j=1}^{k}I(A_1...A_j:A_{j+1}...A_{j+k}:A_{j+k+1})\nonumber\\
&=\sum_{i=1}^n S(A_i| A_{i+1}\dots A_{i+k}) - S(A_1 \dots A_n)\geq 0,
\label{eqn:cycineq}
\end{align}
where  $I(A:B:C)=S(A)+S(B)+S(C)-S(AB)-S(AC)-S(BC)+S(ABC)$ and $S(X\vert Y) = S(XY)-S(Y)$ is the conditional entropy.
This new generalization has also led to the discovery of several further holographic entropy inequalities \cite{BNOSSW}.

Holographic systems are not the only class of quantum systems that obey a more restrictive set of entanglement entropy inequalities. The set of stabilizer states in quantum error correction does so, as well \cite{Stabilizer,gross2013stabilizer,walter_2014}. It is interesting, however, that the known stablizer inequalities are implied by (weaker than) the holographic inequalities, thus suggesting an nontrivial relationship between these two types of states.

It is important to note that the utilization of holography in the inclusion-exclusion proof technique is actually quite minimal. Instead, holographic entropy inequalities are reduced to linear inequalities between the areas of boundaries of certain bulk regions. However, these inequalities are then proved for arbitrary bulk regions, not only for those that minimize the Ryu-Takayanagi entropy. As suggested in \cite{BNOSSW}, it is therefore natural to expect that they hold likewise in condensed matter systems that satisfy an area law. The exploration of this idea, and the applicability of the resulting inequalities for other condensed matter systems, will be the focus of the remainder of this work. As any system with an exact area law entropy scaling necessarily satisfies these entropy inequalities, they may also provide further indication, particularly in the direction of falsification, as to whether a gapped system indeed implies area law scaling for entanglement entropy.

Another potentially interesting relationship here is to the field of AdS/CMT \cite{Hartnoll,McGreevy}; as we will see, the entanglement entropies for gapped phases of matter with an exact area scaling obey the constraints of general holographic systems, which is suggestive of possibly nontrivial holographic duals of condensed matter systems.

\subsection{Organization}
In this work we extend and characterize the realm of applicability 
of the holographic entanglement entropy inequalities to condensed matter systems.
The organization of the paper will be as follows:
In section \ref{sec:trivialtopo}, we formally prove the validity of the cyclic inequalities for systems that obey an exact area law.
In section \ref{sec:TEE}, we show that these inequalities have a valid interpretation as the topological entanglement entropy in two spatial dimensions.
In section \ref{sec:entropy cone}, we fully characterize the entanglement entropy in systems satisfying an exact area law, and we give a minimal and complete set of entropy inequalities and equalities for any fixed number of regions.
We comment about the analogous problem for general quantum system and propose a candidate inequality for four-partite quantum systems.
Finally, we conclude in section \ref{sec:conclusion}.

\section{Gapped systems with trivial topological order}
\label{sec:trivialtopo}

Here we consider gapped systems with trivial topological order in d+1 dimensions. The entanglement entropy $S(A)$ of a subsystem $A$, which measures the entanglement  between $A$ and its complement $A^c$, is defined to be the von Neumann entropy
\begin{equation}
S(A) \equiv -\operatorname{tr} \rho_A \log \rho_A,
\end{equation}
where $\rho_A$ is the reduced density operator for $A$ obtained by tracing out all degrees of freedom outside of $A$ in the many-body ground state of the gapped system.

Note that for a bipartite system in a pure state, the reduced density operators obtained by tracing out either part have the same set of eigenvalues, hence the same von Neumann entropy, which can be seen via a Schmidt decomposition \cite{PreskillLecture} of the pure state that we started with.
This implies that for any division of the system into subsystems $A$ and its complement $A^c$,
\begin{equation}
S(A) = S(A^c)
\label{complement}
\end{equation}
is satisfied.

For gapped systems with trivial topological order, we assume that the entanglement entropy $S(A)$ of a subsystem A scales as the area of its boundary and neglect any sub-area scaling for the moment. Namely,
\begin{equation}\label{eqn:StrictAreaLaw}
S(A)\sim \partial A.
\end{equation}

Note that for any two regions $A$ and $B$ in $d$ spatial dimensions that are non-overlapping except possibly at their boundaries, the entanglement entropy of the combined region $AB$ satisfies
\begin{equation}
S(AB) \sim  \partial(AB) = \partial A + \partial B - 2A\cap B.
\end{equation}
Here $A\cap B$ denotes the area of the codimension 1 hypersurface where regions $A$ and $B$ intersect. The above follows from \eqref{eqn:StrictAreaLaw} because for any $d$-dimensional regions $A$, $B$, and $C$, the triple intersection $A\cap B \cap C$ is of measure zero.

We claim that the cyclic inequality \eqref{eqn:cycineq} for $n = 2k+1$ regions is satisfied as a strict equality in this system, namely
\begin{equation}
\sum_{i=1}^n S(A_i| A_{i+1}\dots A_{i+k}) = S(A_1 \dots A_n),
\label{eqn:arealaw}
\end{equation}
where the sum is cyclic and all indices are taken modulo $n$.
To prove \eqref{eqn:arealaw}, we compute:
\begin{align*}
 \text{LHS} \sim& \sum_{i=1}^n \partial (A_i\dots A_{i+k}) - \partial (A_{i+1}\dots A_{i+k}) \nonumber \\
=&\sum_{i=1}^n \bigg\{ \partial A_i -\sum_{\substack{\{j_1,j_2\} \\ j_1,j_2\in \mathcal{Q}_0}} 2A_{j_1} \cap A_{j_2}
-\sum_{\substack{\{j_1,j_2\} \\ j_1,j_2\in \mathcal{Q}_1}} 2A_{j_1} \cap A_{j_2}\bigg\} \nonumber \\
=& \partial A_1 + \dots +\partial A_n - \sum_{i=1}^n \sum_{ j \in \mathcal{Q}_0} 2A_{i} \cap A_{j}, \nonumber \\
\text{RHS} \sim& \partial(A_1 \dots A_n) \nonumber \\
=& \partial A_1 + \dots +\partial A_n -\sum_{\substack{\{j_1,j_2\} \\ j_1,j_2\in \mathcal{N}}} 2A_{j_1} \cap A_{j_2},
\end{align*}
where $\mathcal{Q}_l= \{i+l,\dots, i+k\}$, $\mathcal{N}=\{1,\dots,n\}$ and $\{j_1,j_2\}$ are unordered pairs that take on the indicated values.
Hence to prove that $\text{LHS} = \text{RHS}$, we need to show that
\begin{align}
\sum_{\substack{\{j_1,j_2\} \\ j_1,j_2\in \mathcal{N}}} A_{j_1} \cap A_{j_2} = \sum_{i=1}^n \sum_{ j \in \mathcal{Q}_0} A_{i} \cap A_{j}.
\label{eqn:intersecting}
\end{align}
The $\text{LHS}$ sums over $n(n+1)/2 = (k+1)(2k+1)$ distinct unordered pairs of indices $\{j_1,j_2\}$, so it contains $(k+1)(2k+1)$ distinct terms.
The summation on the $\text{RHS}$ also contains $(k+1)(2k+1)$ terms. So in order to prove that $\text{LHS} = \text{RHS}$, it suffices to prove that any term appearing in the summation on the $\text{LHS}$ also appears in the summation on the $\text{RHS}$. This is equivalent to proving that for any unordered pair $\{j_1,j_2\}$, $j_1,j_2 \in \{1,\dots, n \}$, there exists $i \in \{1,\dots,n \}, j \in \{i,\dots,i+k \}$, such that $\{j_1,j_2\} = \{i,j\}$.

We have the following two cases:
\begin{enumerate}
\item $j_2 \in \{ j_1,\dots,j_1+k\}$. In this case, we just take $i=j_1, j=j_2$.
\item $j_2 \notin \{ j_1,\dots,j_1+k\}$. In this case, if $j_1 \notin \{j_2,\dots,j_2+k\}$, then $\{j_1,\dots,j_1+k\} \cap \{j_2,\dots,j_2+k\} = \emptyset$, otherwise $j_1+h_1 = j_2+h_2$, for some $h_1,h_2 \in \{1,\dots,k\}$. Therefore, either $j_1 = j_2 + h_2 -h_1$, or $j_2 = j_1 +h_1-h_2$. Since either $h_1-h_2 \in \{1,\dots,k\}$, or $h_2-h_1 \in \{1,\dots,k\}$, we are forced to conclude that either $j_2 \in \{ j_1,\dots,j_1+k\}$, or $j_1 \in \{j_2,\dots,j_2+k\}$. Both lead to contradictions. So if $j_2 \notin \{ j_1,\dots,j_1+k\}$, and if $j_1 \notin \{j_2,\dots,j_2+k\}$, then $\{j_1,\dots,j_1+k\} \cap \{j_2,\dots,j_2+k\} = \emptyset$. Since each set contains $k+1$ distinct numbers, and if their intersection is empty, their union would contain $2k+2$ distinct numbers, contradicting $n=2k+1$. Thus we finally arrive at the conclusion that $j_2 \notin \{ j_1,\dots,j_1+k\} \implies j_1 \in \{j_2,\dots,j_2+k\}$. In this case, we just take $i = j_2, j = j_1$.
\end{enumerate}
Combining cases 1 and 2, we conclude that for any unordered pair $\{j_1,j_2\}$, $j_1,j_2 \in \{1,\dots, n \}$, there exists $i \in \{1,\dots,n \}, j \in \{i,\dots,i+k \}$, such that $\{j_1,j_2\} = \{i,j\}$. Hence \eqref{eqn:intersecting} indeed holds
and \eqref{eqn:arealaw} is exactly satisfied by such systems.

In section~\ref{sec:entropy cone} below, we will generalize this result and identify \emph{all} entropy inequalities and equalities that are obeyed in systems with an exact area-law scaling.
We will find that any holographic entropy inequality is valid for systems with an exact area law.

\section{Topological entanglement entropy}
\label{sec:TEE}
\subsection{Construction and validity}
\label{subsec:generalconst}

Here we consider general gapped systems in $2+1$ dimensions. It is shown in \cite{TEEKitaevPreskill,TEEWenLevin,3dTEE} that the entanglement entropy of a region $A$ with smooth boundary has the form

\begin{equation}
S_A = \alpha L - b_0\gamma
\label{enscaling}
\end{equation}
in the limit $L/a \rightarrow \infty$ where $a$ is the correlation length. Here $b_0$ denotes the number of connected components of $\partial A$, the boundary of region $A$. The topological entanglement entropy $-\gamma$ is a universal constant characterizing the topological state and $\alpha$ is a non-universal and ultraviolet divergent coefficient dependent on the short wavelength modes near the boundary of region $A$. In particular, $\gamma=\log \mathcal{D}$ captures the far-IR behaviour of entanglement and the total quantum dimension, $\mathcal{D}$, which can be obtained from topological quantum field theory computations, is related to the number of superselection sectors of the system\cite{TEEKitaevPreskill}.

We divide the plane into $2k+2$ regions, labeled by $A_0,A_1,A_2,\dots,A_{2k+1}$, where $A_0$ labels the complement of $A_1A_2\dots A_{2k+1}$, i.e., $A_0\equiv (A_1A_2\dots A_{2k+1})^c$. In order for the topological entropy $S_{\text{topo}}$ defined in \eqref{tee} to be a topological invariant, we require the division of the plane to satisfy
\begin{equation}
\bigcap_{i\in I} A_i = \emptyset, \ \text{for all} \ I \subset \{0,1,2,\cdots,2k+1 \}, \text{such that} \ 0 \in I, \text{and} \ \lvert I \rvert > 3.
\label{cond}
\end{equation}
In other words, there is no point on the plane that is shared by $A_0$ and more than two other regions.
We define the topological entropy $S_{\text{topo}}$ for $2k+2$ regions as

\begin{equation}
S_{\text{topo}} \equiv \sum_{i=1}^{2k+1} S(A_i| A_{i+1}\dots A_{i+k}) - S(A_1 \dots A_{2k+1}),
\label{tee}
\end{equation}
where all indices are taken modulo $(2k+1)$. Note that our definition of the topological entropy reduces to the Kitaev-Preskill one \cite{TEEKitaevPreskill} when $k=1$ (i.e., a division of the plane into $4$ regions). Also note that our calculation in section \ref{sec:trivialtopo} implies that for gapped systems with trivial topological order, $S_{\text{topo}}=0$, that is, the dependence of $S_{\text{topo}}$ on the length of the boundaries cancels out.

To see that $S_{\text{topo}}$ is a topological invariant, consider deforming the boundary between regions labeled by the index set $J \subset \{0, 1,2,\dots,2k+1\}$, i.e., points in the set
\begin{equation*}
S \equiv \bigcap_{j \in J} A_j.
\end{equation*}

We note the following properties of the entanglement entropy before proceeding to the main arguments:
\begin{itemize}
\item Under general deformations of the plane, the change in the entanglement entropy of any region $A$, $\Delta S(A)$, is equal to the change in the entanglement entropy of the complement of $A$, $\Delta S(A^c)$. This follows as a consequence of \eqref{complement}.

\item For any $k\notin J$, points of deformation (points in $S$) are far from $A_k$. Therefore, we expect $\Delta S(A_k)=0$, provided that all regions are large compared to the correlation length. In the same spirit, $\Delta S(A)=0$ for any region $A$ that is a union of such $A_k$'s. It then follows if $A$ is appended to any region $B$, the change in entanglement of that region is unaffected, namely, $\Delta S(B\cup A)=\Delta S(B) $.

\end{itemize}

Now we argue for the topological invariance of $S_{\text{topo}}$. Possible deformations of the regions are classified into the following two cases:
\begin{enumerate}
\item \label{case1} $0 \notin J$. In this case,
\begin{align*}
\Delta S_{\text{topo}} &= \sum_{i=1}^{2k+1} [\Delta S(A_i A_{i+1}\dots A_{i+k}) - \Delta S(A_{i+1}\dots A_{i+k})]- \Delta S(A_1 \dots A_{2k+1}) \nonumber \\
&= \sum_{i=1}^{2k+1} [\Delta S((A_i A_{i+1}\dots A_{i+k})^c)- \Delta S(A_{i+1}\dots A_{i+k})] -  \Delta S((A_1 \dots A_{2k+1})^c) \nonumber \\
&= \sum_{i=1}^{2k+1} [\Delta S(A_0A_{i+1+k} A_{i+2+k}\dots A_{i+2k})- \Delta S(A_{i+1}\dots A_{i+k})] -  \Delta S(A_0) \nonumber \\
&=\sum_{i=1}^{2k+1} \Delta S(A_{i+1+k} A_{i+2+k}\dots A_{i+2k})- \sum_{i=1}^{2k+1} \Delta S(A_{i+1}\dots A_{i+k}) \nonumber \\
&=\sum_{i=1}^{2k+1} \Delta S(A_{i+1} A_{i+2}\dots A_{i+k})- \sum_{i=1}^{2k+1} \Delta S(A_{i+1}\dots A_{i+k}) = 0,
\end{align*}
where in the last step, we cyclically left permute the summands in the first summation by $k$ steps, which leaves the summation invariant.

\item $0 \in J$. In this case, by  condition\eqref{cond}, we can either have $\lvert J \rvert = 2$ or $\lvert J \rvert = 3$.

\begin{enumerate}
\item \label{case2(a)} $\lvert J \rvert = 2$. Denote the only nonzero element in $J$ as $j$.
In this case,
\begin{align*}
\Delta S_{\text{topo}} =& \Delta S(A_j A_{j+1} \dots A_{j+k}) + \Delta S(A_{j-1} A_{j} \dots A_{j+k-1}) + \cdots \nonumber \\
& + \Delta S(A_{j-k} A_{j-k+1} \dots A_{j}) - \Delta S(A_j A_{j+1} \dots A_{j+k-1}) \nonumber \\
& - \Delta S(A_{j-1} A_{j} \dots A_{j+k-2}) - \cdots - \Delta S(A_{j-k+1} A_{j-k+2} \dots A_{j}) \nonumber \\
& -\Delta S(A_1A_2\dots A_{2k+1}) \nonumber \\
=& (k+1)\Delta S(A_j) - k\Delta S(A_j) -\Delta S(A_j) = 0.
\end{align*}

\item \label{case2(b)} $\lvert J \rvert = 3$. Denote the nonzero elements in $J$ as $j_1$ and $j_2$. Moreover, since $j_1, j_2 \in \{1,2,\cdots, 2k+1 \}$,  $|j_1-j_2|\leq k$, i.e., they are separated by a distance of at most $k$ (note $j_1, j_2$ are $\text{mod} \ (2k+1)$ integers). Without loss of generality, we assume $j_2 = j_1 + l$, for some $0 < l \leq k$. We further write $j_1$ as $j$ for simplicity.
In this case,
\begin{align*}
\Delta S_{\text{topo}} =& \Delta S(A_j \dots A_{j+l} \dots A_{j+k}) + \cdots +\Delta S(A_{j+l-k} \dots A_{j} \dots A_{j+l}) \nonumber \\
&+\Delta S(A_{j+l-k-1}\dots A_{j} \dots A_{j+l-1}) + \cdots + \Delta S(A_{j-k}A_{j-k+1}\dots A_{j}) \nonumber \\
&+ \Delta S(A_{j+l} A_{j+l+1}\dots A_{j+l+k}) + \cdots + \Delta S(A_{j+1} \dots A_{j+l} \dots A_{j+k+1}) \nonumber \\
&- \Delta S(A_j \dots A_{j+l} \dots A_{j+k-1}) - \cdots -\Delta S(A_{j+l-k+1} \dots A_{j} \dots A_{j+l}) \nonumber \\
&-\Delta S(A_{j+l-k}\dots A_{j} \dots A_{j+l-1})  - \cdots - \Delta S(A_{j-k+1}A_{j-k+2}\dots A_{j}) \nonumber \\
&- \Delta S(A_{j+l} A_{j+l+1}\dots A_{j+l+k-1}) - \cdots - \Delta S(A_{j+1} \dots A_{j+l} \dots A_{j+k}) \nonumber \\
&- \Delta S(A_1A_2\dots A_{2k+1}) \nonumber \\
=& (k-l+1)\Delta S(A_j A_{j+l}) + l\Delta S(A_j) + l\Delta S(A_{j+l}) \nonumber \\
&- (k-l)\Delta S(A_j A_{j+l})- l\Delta S(A_j) - l\Delta S(A_{j+l}) - \Delta S(A_jA_{j+l})=0.
\end{align*}
\end{enumerate}
\end{enumerate}

To see that $S_{\text{topo}}$ is a universal quantity, we consider the same argument in \cite{TEEKitaevPreskill}, where a smooth deformation of the local Hamiltonian during which no quantum critical points are encountered.
Since the Hamiltonian is local, any smooth deformations of the Hamiltonian can be written as a sum of smooth deformations of local terms. Moreover, by utilizing the topological invariance of $S_{\text{topo}}$, we may deform the regions in the following ways while keeping $S_{\text{topo}}$ invariant:
\begin{itemize}
\item Stretch the boundaries of the regions so that $L \rightarrow \infty$, and the entanglement entropy of a region takes the form of \eqref{enscaling}.
\item Deform the boundaries of the regions so that all deformation of the Hamiltonian happens locally in the bulk of the regions.
\end{itemize}
We further assume that the correlation length remains small compared to the size of the regions throughout the deformation. Hence any local deformations of the Hamiltonian in the bulk only has miniscule effects for the ground state near the boundary.
As a result, the entanglement entropies of the deformed regions (and hence $S_{\text{topo}}$) should not be affected by such local deformations of the Hamiltonian.

Thus we conclude that the topological entropy we defined in \eqref{tee} is both a topological invariant (invariant under deformations of the boundary of the regions that keep the topology of the regions unchanged) and a universal quantity (invariant under smooth deformations of the Hamiltonian during which no quantum critical points are encountered).

For a general division of the plane into $2k+2$ regions that satisfies condition \eqref{cond}, we can compute the topological entropy:
\begin{align}
S_{\text{topo}} &=  -\gamma\Big\{\sum_{i=1}^{2k+1}\Big(b_0[\partial (A_i\dots A_{i+k})] - b_0[\partial (A_{i+1}\dots A_{i+k})]\Big) + b_0[\partial (A_1\dots A_{2k+1})]\Big\},
\label{eqn:generalreg}
\end{align}
where all indices are taken modulo $(2k+1)$, and $b_0[\partial A]$ denotes the zeroth Betti number (the number of connected components) of the boundary of a region $A$.

Hence $S_{\text{topo}}$ is proportional to the topological entanglement entropy $-\gamma$, with the proportionality constant determined by the topology of the regions. This implies that we can extract the topological entanglement entropy of a $2+1$ dimensional topologically ordered system with a mass gap by suitably divide the system into $2k+2$ regions, and compute the topological entropy $S_{\text{topo}}$.

\subsection{Examples}

Here we consider a few examples which elucidate some of the general constructions in section~\ref{subsec:generalconst}. Figures~\ref{fig:pie} and \ref{fig:piehole} illustrate two possible divisions of the plane into $2k+2$ regions that satisfy condition \eqref{cond}.

There are three types of deformations to the regions for both divisions. We consider the change in $S_{\text{topo}}$ under such deformations.

First, consider deforming the boundary between two regions in figures~\ref{fig:pie} and \ref{fig:piehole}. This can either be the boundary between two slices of the pie, say $A_1$ and $A_2$, or the boundary between a slice of the pie and the complement of the pie, say $A_1$ and $A_0$. In the former case, region $A_0$ is not involved in the deformation, so case~\ref{case1} of our general analysis for the topological invariance of $S_{\text{topo}}$ implies that $\Delta S_{\text{topo}}=0$. In the latter case, there is one more region ($A_1$) besides $A_0$ that are involved in the deformation, so case~\ref{case2(a)} of our general analysis implies $\Delta S_{\text{topo}}=0$.

Next, consider deforming a triple point where three regions meet. Without loss of generality, consider the three regions $A_1$, $A_2$ and $A_0$. There are two more regions ($A_1$ and $A_2$) besides $A_0$ that are involved in the deformation, so case~\ref{case2(b)} of our general analysis implies that $\Delta S_{\text{topo}}=0$.

Note that for the division in figure~\ref{fig:pie}, one could also deform the center of the pie chart, which is seemingly different from other points in the plane. All $2k+1$ regions but $A_0$ are involved in the deformation. However, since $A_0$ is not involved, case~\ref{case1} of our general analysis still applies in this case, and $\Delta S_{\text{topo}} = 0$.

To compute the topological entropy $S_{\text{topo}}$ for these two divisions, we apply the general formula \eqref{eqn:generalreg}. For pie-chart divisions in figures~\ref{fig:pie} and \ref{fig:piehole},
\begin{equation*}
b_0[\partial (A_i\dots A_{i+k})] = b_0[\partial (A_{i+1}\dots A_{i+k})] = 1,
\end{equation*}
therefore the summation in \eqref{eqn:generalreg} gives zero, and one simply counts $b_0$ for the boundary of the union of all $2k+1$ regions, which yields $-\gamma$ and $-2\gamma$ respectively.

\begin{figure}[h]
    \centering
    \includegraphics[width=0.5\textwidth]{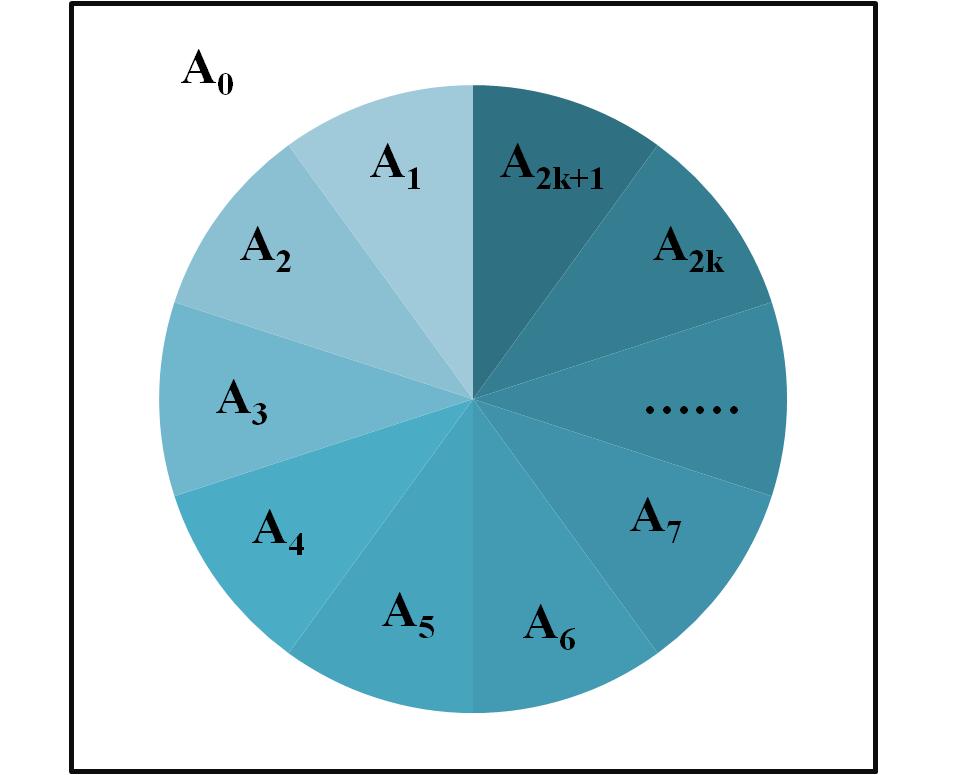}
    \caption{A pie-chart division of the plane into $2k+2$ regions, labeled by $A_0, A_1, \dots, A_{2k+1}$. }
    \label{fig:pie}
\end{figure}

\begin{figure}[ht]
    \centering
    \includegraphics[width=0.5\textwidth]{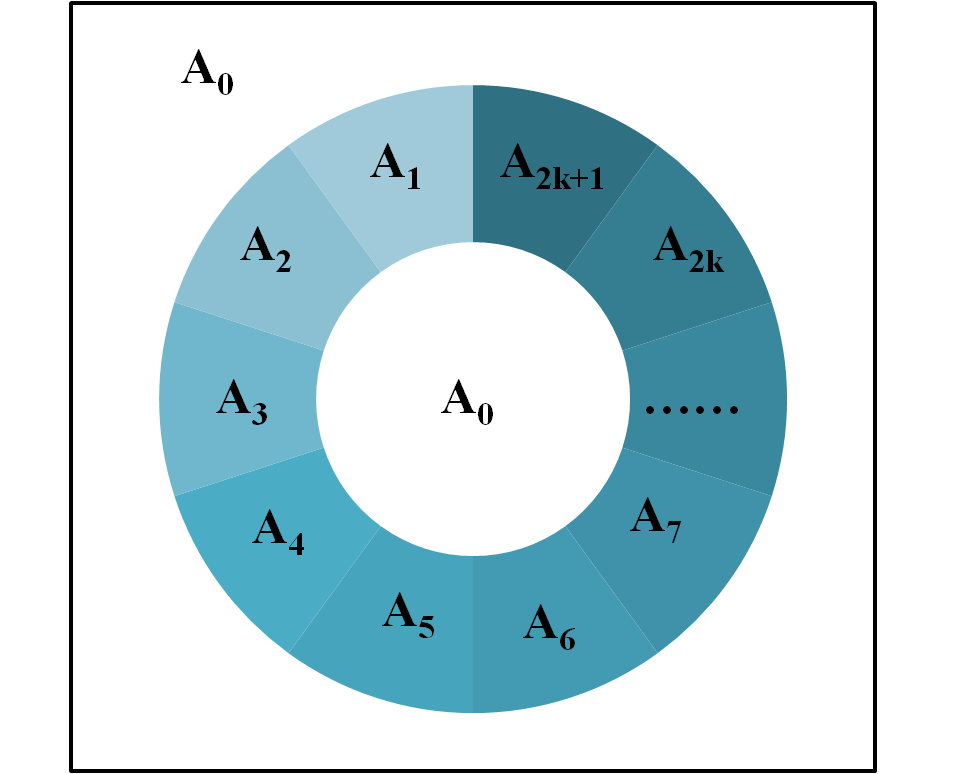}
    \caption{A pie-chart division of the plane into $2k+2$ regions, with a hole in the middle, labeled by $A_0, A_1, \dots, A_{2k+1}$. }
    \label{fig:piehole}
\end{figure}

\begin{figure}[ht]
    \centering
    \includegraphics[width=0.5\textwidth]{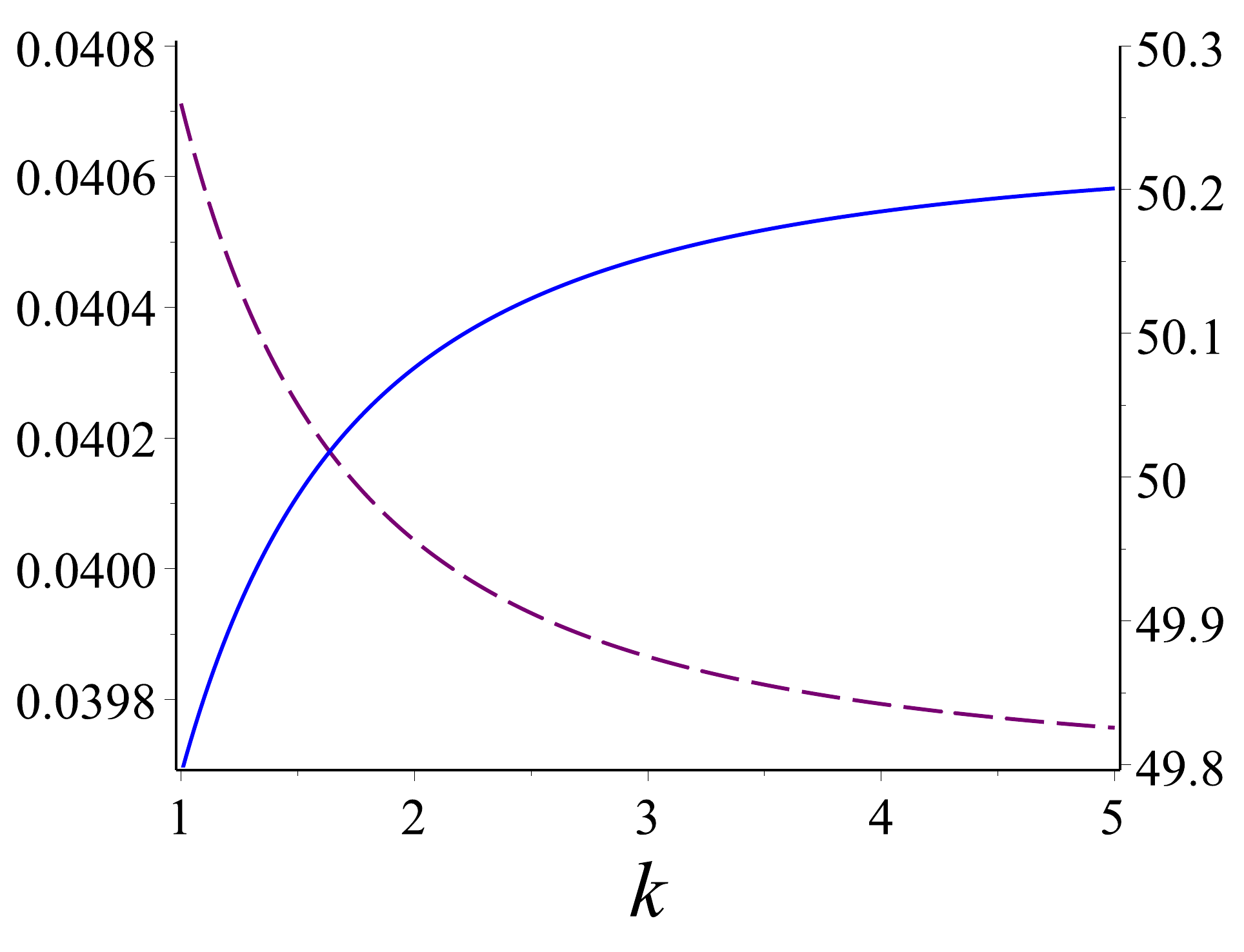}
    \caption{For systems that deviate from exact area law, the general behaviors of corrections  to $S_{topo}$  in the forms of  $1/\ell$ and $\ell\log\ell$ are sketched in purple and blue respectively (color online) for $R=10$. The numerical values are up to some unknown constant of order $\epsilon$ or $\beta$. }
    \label{fig:devplot}
\end{figure}

\subsection{Beyond area-law scaling}
In the above sections, we have considered systems where the entanglement entropy is in the form of \eqref{enscaling}. In general, \eqref{enscaling} will be supplemented with various corrections. We here consider a few examples and examine the behavior of the topological entropy \eqref{tee} as a function of $k$ in systems that deviates from exact area-law scaling.

Since the correction leads to imperfect cancellation of the local contributions to entanglement entropy, \eqref{tee} thus yields the topological entropy up to some local correction factors. For the following examples, we restrict ourselves to the pie-chart division in $2+1$ dimensions in figure~\ref{fig:pie} with radius $R$ in units of the lattice size $a$. For the sake of simplicity, the $2k+1$ slices are divided evenly.

For a generic $2+1$ dimensional gapped system with non-trivial topological order, the entanglement entropy of some region A with perimeter $L$ large compared to the correlation length is given by
\begin{equation}
S_A = \alpha L - b_0 \gamma +\frac{\beta_1 }{L}+\frac{\beta_3 }{L^3} + \cdots
\end{equation}
where the correction from \eqref{enscaling} assumes the form of $\beta_p/L^{p}$, for all odd integers $p$.

The $\mathcal{O}(L^{-p})$ correction to the topological entropy \eqref{tee} is given by
\begin{equation}
\Delta S_{\text{topo}}^p =\frac{\beta_p}{(2\pi R)^p} \Big\{ (2k+1)^{p+1}\Big(\frac{1}{(x+1)^p} -\frac{1}{x^p}\Big)-1\Big\},
\end{equation}

where $x=(1+2/\pi)k+1/\pi$.
For $k\geq1$, $(\partial |\Delta S^p_{\text{topo}}|/\partial k)<0$ for each $p$. The higher $k$ expressions are therefore slightly less sensitive to deviations of the local piece from perfect area-law scaling.

More generally, we may consider entropy scaling $S_A=\alpha L-\gamma+\epsilon f(L/a)$ where we recast other small deviations in the local piece of entanglement entropy into the form of $f(L/a)$. Let $\ell=L/a$ for a lattice with spacing $a$, we here briefly sketch the behaviors for corrections $f(\ell)=\log(\ell)$  and $f(\ell)=\ell \log(\ell)$ \cite{AreaLawEntRev,FreeFermion}. Note that near criticality\cite{FreeFermion}, $\ell\log(\ell)$ scaling becomes dominant in the local piece of entanglement entropy. For systems sufficiently far from a phase transition, we can treat them as an order $\epsilon$ correction in the area-law systems.

For $f(\ell)=\log(\ell)$,
\begin{equation}
\Delta S_{\text{topo}}=\epsilon\log \Big\{\frac{1}{2\pi R}\Big(1+\frac{1}{x} \Big)^{2k+1} \Big\},
\end{equation}
where $x=(1+2/\pi)k+1/\pi$.

And for $f(\ell)=\ell\log(\ell)$,
\begin{equation}
\Delta S_{\text{topo}}=\epsilon\log \Big \{ \frac{1}{(2\pi R)^{2\pi R}} \Big( \frac{(2R + \frac{2\pi}{2k+1}(k+1)R)^{2R + \frac{2\pi}{2k+1}(k+1)R}}{(2R + \frac{2\pi}{2k+1}kR)^{2R + \frac{2\pi}{2k+1}kR}} \Big)^{2k+1} \Big\}.
\end{equation}
In such cases, $(\partial |\Delta S_{\text{topo}}|/\partial k)>0$ for regions with boundary much larger than the correlation length. which renders higher $k$ definitions of the topological entropy $S_{\text{topo}}$ more sensitive to local corrections to entanglement entropy.
In principle, for some realistic condensed matter system with non-trivial topological order, the generalized definition \eqref{tee} offers a wider range of selection where one can choose the optimal $k$ for purposes of studying both the topological entanglement entropy and local deviations from area-law scaling.

\section{All entropy inequalities for systems with an exact area law}
\label{sec:entropy cone}
We will now derive the full set of constraints satisfied by the entanglement entropy in systems with an exact area law, $S(A) \sim \partial A$.
We begin with a useful construction from \cite{BNOSSW} that allows us to reduce from continuous geometries to a combinatorial problem.
For simplicity, we shall assume that the system lives on a manifold without boundary.

Let $A_1, \dots, A_n$ be disjoint (apart from their boundaries) regions in a system with an exact area law.
We introduce the purifying region $A_{n+1}$ as the closure of $(A_1 \cup \dots \cup A_n)^c$.
The entropy of an arbitrary composite region $A_I = \bigcup_{i \in I} A_i$ for $I \subseteq [n+1] := \{1,\dots,n+1\}$ can then be evaluated in the following way,
\[ S(A_I) \sim \partial A_I = \sum_{i \in I, j \in I^c} \partial A_i \cap \partial A_j, \]
where $I^c$ denotes the complement of $I$ in $[n+1]$.

We now consider the undirected complete graph on $n+1$ vertices, equipped with the edge weights $w(i,j) = \partial A_i \cap \partial A_j$.
Let $w(I,J) = \sum_{i \in I, j \in J} w(i,j)$ denote the total weight of all edges between two disjoint subsets $I$ and $J$, and $\delta(I) = w(I,I^c)$ the \emph{cut function}.
Then it follows from the above that
\[ S(A_I) \sim \partial A_I = \delta(I). \]
Conversely, for any given undirected graph with non-negative edge weights we can always construct a geometry and associated regions $A_1, \dots, A_{n+1}$ such that $\partial A_I = \delta(I)$ (cf.\ \cite{BNOSSW}).
Therefore, proving entropy inequalities for systems with an exact area law is completely equivalent to proving linear inequalities for the cut function.

To determine if a linear inequality
\begin{equation}
\label{eq:cut inequality}
  \sum_I c_I \, \delta(I) \geq 0
\end{equation}
holds for the cut function in an arbitrary undirected weighted graph, we expand:
\begin{align*}
  &\sum_I c_I\, \delta(I)
= \sum_I c_I\, w(I,I^c)
= \sum_I c_I \, \sum_{\mathclap{i \in I, j \in I^c}} w(i,j)
= \sum_{i \neq j} w(i,j) \, \sum_{\mathclap{I : i \in I, j \not\in I}} c_I
= \sum_{\{i,j\}} w(i,j) \, \sum_{\mathclap{I : i \in I \text{ xor } j \in I}} c_I
\end{align*}
In the last step, the outer sum is over edges of the undirected graph.
Since the edge weights $w(i,j)$ are arbitrary non-negative numbers, this immediately implies that the inequality \eqref{eq:cut inequality} is valid if and only if
\begin{equation}
\label{eq:single edge}
\end{equation}
\[ \sum_{I : i \in I \text{ xor } j \in I} c_I \geq 0 \]
for any edge $\{i,j\}$.
Note that \eqref{eq:single edge} asserts simply that the inequality \eqref{eq:cut inequality} holds for the graph with a single edge $\{i,j\}$ of edge weight $1$, since the cut function in this case is given by
\begin{equation}
\label{eq:bell entropies}
(\beta_{ij})_I := \delta(I) = \begin{cases} 1 & i \in I \text{ xor } j \in I \\ 0 & \text{otherwise} \end{cases}.
\end{equation}
But note that $\beta_{ij}$ are precisely the entropies of a Bell pair shared between subsystems $i$ and $j$ in an $(n+1)$-partite pure state.
In view of our reduction from systems with an exact area law to graphs, we thus obtain the following result:

\begin{lem}
\label{lem:ieq criterion}
  An entropy inequality $\sum_{I \subseteq [n]} c_I \, S(A_I) \geq 0$ is valid for all systems with an exact area law if and only if it is valid for the entropies of Bell pairs shared between any two subsystems $A_i$ and $A_j$ of the purified $(n+1)$-partite system.
\end{lem}

In section~\ref{sec:trivialtopo}, we had proved that the cyclic inequalities \eqref{eqn:cycineq} hold with \emph{equality} for system that satisfy an exact area law.
It follows immediately from lemma~\ref{lem:ieq criterion} that we can test the validity of an arbitrary entropy equality by verifying that they hold with equality when evaluated for Bell pairs.
In \cite{BNOSSW}, it was observed that this is the case for the cyclic inequalities \eqref{eqn:cycineq} as well as four other holographic entropy inequalities established therein.
It follows that all these inequalities hold with equality for systems satisfying an exact area law.
In particular, this confirms our explicit derivation for the cyclic inequalities in section~\ref{sec:trivialtopo}.

\bigskip

In principle, lemma~\ref{lem:ieq criterion} solves completely the problem of characterizing the entanglement entropy in systems with an exact area law.
We will now describe the set of all possible entanglement entropies more concretely.
For this, it is useful to observe that, for any fixed number of regions $n$, the collection of valid entropy inequalities $\sum_I c_I S(A_I) \geq 0$ cuts out a convex cone.
This cone consists of all vectors $s = (S(A_I))_{\emptyset \neq I \subseteq [n]} \in \RR^{2^n-1}$ formed from the entanglement entropies obtained by varying over arbitrary regions $A_1, \dots, A_n$ and all systems satisfying an exact area law.
Following \cite{EntCone1,EntCone2,gross2013stabilizer,BNOSSW}, we shall call it the \emph{area-law entropy cone}.
Like any convex cone, it can be dually described in terms of its extreme rays, which we obtain immediately from lemma~\ref{lem:ieq criterion}:

\begin{lem}
  The extreme rays of the area-law entropy cone for $n$ regions are given by the entropy vectors $\beta_{ij}$ of Bell pairs shared between any two subsystems in the purified $(n+1)$-partite system.
\end{lem}

Since the entropies of Bell pairs can be realized holographically, we may think of the area-law entropy cone as a degeneration of the holographic entropy cone defined in \cite{BNOSSW}.
It is arguably the smallest entropy cone that can capture bipartite entanglement.

The following theorem then gives a complete characterization of the entanglement entropy in systems with an exact area law:

\begin{thm}
\label{thm:complete}
  A minimal and complete set of entropy (in)equalities for systems with an exact area law is given by
  (1) the subadditivity inequality $S(A_1) + S(A_2) \geq S(A_1 A_2)$ and its permutations,
  (2) the Araki-Lieb inequality $S(A_1) + S(A_1 \dots A_n) \geq S(A_2 \dots A_n)$ and its permutations, and
  (3) the multivariate information \emph{equalities} $\sum_{I \subseteq V} (-1)^{\lvert I \rvert} S(A_I) = 0$ induced by any subset $V \subseteq [n]$ of cardinality at least three.
\end{thm}
\begin{proof}
  We first argue that the entropy equalities in (3) are correct and linearly independent.
  Their correctness can be verified by evaluating them on the rays $\beta_{ij}$ for $i < j \in [n+1]$:
  \[
    \sum_{I \subseteq V} (-1)^{\lvert I \rvert} (\beta_{ij})_I
    = \sum_{\mathclap{I : i \in I \subseteq V \setminus \{j\}}} (-1)^{\lvert I \rvert}
    + \sum_{\mathclap{I : j \in I \subseteq V \setminus \{i\}}} (-1)^{\lvert I \rvert}
  \]
  By symmetry, it suffices to consider the first sum. If $i \not\in V$ then it is zero. Otherwise,
  \[
    \sum_{\mathclap{I : i \in I \subseteq V \setminus \{j\}}} (-1)^{\lvert I \rvert}
    = - \sum_{\mathclap{J : J \subseteq V \setminus \{i,j\}}} (-1)^{\lvert J \rvert}
    = - \sum_{k=0}^{\lvert V \setminus \{i,j\} \rvert} (-1)^k \binom {\lvert V \setminus \{i,j\} \rvert} k
    = 0
  \]
  by the standard identity for an alternating sum of binomial coefficients, which is applicable since $\lvert V \setminus \{i,j\}\rvert \geq 3-2=1$.
  The fact that the equalities in (3) are all linearly independent can easily be seen by induction on $\lvert V \rvert$.

  It follows from the above that the area-law cone is contained in a linear subspace of dimension
  $2^n - 1 - \sum_{k=3}^n \binom n k = n + \binom n 2 = \binom {n+1} 2.$
  We will now show that this is indeed the dimension of the area-law cone.
  For this, it suffices to observe that $(\beta_{ij})_k + (\beta_{ij})_l - (\beta_{ij})_{\{k,l\}} = 2$ if $\{i,j\} = \{k,l\}$, and otherwise zero.
  This not only implies that the $\binom {n+1} 2$ many extreme rays $\beta_{ij}$ are all linearly independent, but also that the area-law entropy cone is cut out by the inequalities
  \[ S(A_k) + S(A_l) \geq S(A_{kl}) \]
  on the subspace defined by the multivariate information equalities (3).
  For $l \leq n$, these are just the inequalities in (1), while for $l = n+1$ we obtain the inequalities in (2) by using the relation $S(A_I) = S(A_{I^c})$.
  It is clear from the above that the entropy (in)equalities (1)--(3) form a minimal set.
\end{proof}

Geometrically, the area-law entropy cone is an ``orthant'' of dimension $\binom {n+1} 2$, as follows from the proof of the theorem, where we have shown that the extreme rays are linearly independent. We note that the set of defining (in)equalities is in general not unique as the entropy cone has positive codimension for $n > 2$.

\subsection{Generating entropy equalities from graphs}
\label{sec:graph}

Our method can be easily adapted to include information about the spatial connectivity of the regions $A_1, \dots, A_n$ that enter an (in)equality: If we can guarantee that $A_i \cap A_j = \emptyset$ then we do not need to consider the corresponding Bell pair $\beta_{ij}$ when verifying an entropy (in)equality by using lemma~\ref{lem:ieq criterion}.
For a concrete example, consider the conditional mutual information $I(A:C|B) = S(AB) + S(BC) - S(ABC) - S(B)$, which is equal to zero for all Bell pairs except for $\beta_{AC}$.
If we choose $A$, $B$, and $C$ as in the figure~\ref{fig:cmi regions} below then $I(A:C|B) = 0$ for systems with an exact area law, since $A \cap C = \emptyset$. This cancellation has been used in \cite{TEEWenLevin} to extract the topological entanglement entropy.

\begin{figure}[h]
\centering
\includegraphics[width=0.3\textwidth]{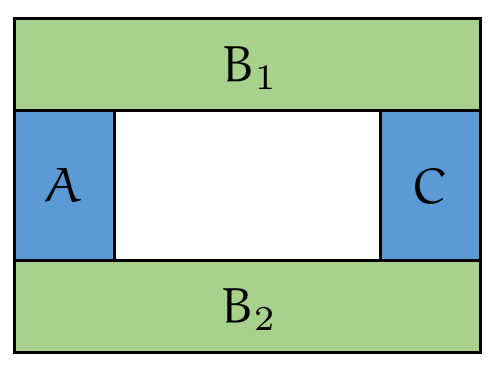}
\caption{In systems with an exact area law, the conditional information $I(A:C|B)$ vanishes for this configuration of regions \cite{TEEWenLevin}.}
\label{fig:cmi regions}
\end{figure}

This example has the following pleasant generalization:

\begin{lem}
\label{lem:graph eqn}
Let $(V,E)$ be an undirected graph on the vertex set $V = [n]$, $n \geq 3$, and let $A_1, \dots, A_n$ be regions such that $\partial A_i \cap \partial A_j \neq \emptyset$ only if $\{i,j\} \in E$.
Then we have the following entropy equality,
\begin{equation}
\label{eq:graph eqn}
  \sum_{I \subseteq V} (-1)^{\lvert I \rvert} \sum_{J \in \pi_0(I)} S(J) = 0,
\end{equation}
where $\pi_0(I)$ denotes the connected components of the induced subgraph with vertex set $I$.

For a complete graph $(V,E)$, \eqref{eq:graph eqn} is precisely one of the multivariate information equalities proved in theorem~\ref{thm:complete} for arbitrary regions.
\end{lem}
\begin{proof}

  It suffices to argue that the difference to the multivariate information vanishes given our assumption on the spatial connectivity of the regions $A_1, \dots, A_n$.
  For this, note that, for any $I \subseteq V$,
  \[ \left( \sum_{J \in \pi_0(I)} S(J) \right) - S(I) = \sum_{J \neq K \in \pi_0(I)} \sum_{j \in J, k \in K} \partial A_j \cap \partial A_k = 0, \]
  since $j$ and $k$ are in different connected components so that $(j,k) \not\in E$ and therefore $\partial A_j \cap \partial A_k = \emptyset$ by our assumption.
\end{proof}

For the graph displayed in figure~\ref{fig:strongsubadd} below we recover the statement derived above that $I(A:C|B) = 0$ for systems with $A \cap C = \emptyset$.

\begin{figure}[h]
\centering
\begin{tikzpicture}[%
  every node/.style={draw,fill=black,circle,minimum size=0.3cm},node distance=0.75cm]
  \node[circle, label=$A$] (one) at (0,0) {};
  \node[right=of one,label =$B$] (two) {};
  \node[right=of two, label=$C$] (three) {};

  \draw (one) -- (two) -- (three);
\end{tikzpicture}
\caption{Graph corresponding to $I(A:C|B) = 0$ for regions with $A \cap C = \emptyset$.}
\label{fig:strongsubadd}
\end{figure}
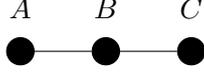

We note that the connectivity assumption in lemma~\ref{lem:graph eqn} is equivalent to requiring that $I(A_i : A_j) = 0$ for all $\{i,j\} \not\in E$.
We may therefore think of \ref{eq:graph eqn} as a \emph{constrained entropy equality} in the sense of \cite{LindenWinter}.
Below we list all non-trivial constrained entropy equalities obtained by lemma~\ref{lem:graph eqn} from graphs with four vertices:
\begin{align*}
 \text{(e.g., $K_{1,3}$):}\quad&S(A_1 A_2 A_3 A_4)-S(A_2 A_3 A_4)-S(A_1 A_2 A_4)-S(A_1 A_3 A_4)\\
 +&S(A_1 A_4)+S(A_2 A_4)+S(A_3 A_4)-S(A_4)= 0 \\
 \text{($C_4=K_{2,2}$):}\quad&S(A_1 A_2 A_3 A_4)-\sum_{i=1}^{4}S(A_i|A_{i+1}A_{i+2})=0  \\
 \text{(Diamond):}\quad&S(A_1 A_2 A_3 A_4)-S(A_1 A_2 A_3)-S(A_2 A_3 A_4)-S(A_3 A_4 A_1)\\
 -&S(A_1 A_2 A_4)+S(A_1 A_2)+S(A_1 A_4)+S(A_2 A_3)\\
 +&S(A_2 A_4)+S(A_3 A_4)-S(A_2)-S(A_4)=0  \\
 \text{($K_4=W_4$):}\quad&S(A_1 A_2 A_3 A_4) - S(A_1A_2A_3) - S(A_1A_2A_4) - S(A_1A_3A_4) - S(A_2A_3A_4) \\
 +&S(A_1A_2)+S(A_1A_3)+S(A_1A_4)+S(A_2A_3)+S(A_2A_4)+S(A_3A_4)\\
 -&S(A_1)-S(A_2)-S(A_3)-S(A_4)=0
\end{align*}
The last equality is in fact unconditionally true for system with an exact area law; it is the fourpartite information equality from theorem~\ref{thm:complete}.

\subsection{The search for general quantum entropy inequalities}
The simplistic method of graph combinatorics has thus far managed to reproduce the forms of many familiar entropic inequalities. In light of the observation, one may suspect that it could also be useful in generating other generic $n$-party quantum inequalities. In light of the above speculation, we attempt a simple search for 4-party linear quantum inequalities beyond the \emph{von Neumann entropy inequalities} \eqref{eqn:SA}--\eqref{eqn:SSA}.

To search for the 4-party non-von Neumann inequality candidates, we employ again the notion of an entropy cone \cite{EntCone1,EntCone2}. Let $K=\{A, B, C, D\}$ be the labels of the 4 systems and E be the purifying system. Given a density operator $\rho_{ABCD}$, we obtain the entropy vector $\mathbf{v}=(v_{I})\in\mathbb{R}^{15}$ where $\emptyset\neq I\subseteq K$ and $v_{I} = S(I)_\rho$ denotes the von Neumann entropy of the reduced density matrix $\rho_I$ obtained by tracing out all but the systems in $I$. Consider the set $\Gamma^*_4\subset\mathbb{R}^{15}$ of all entropy vectors $\mathbf{v}$ produced by physical 4-party quantum states, we define the 4-party \emph{quantum entropy cone} as the closure $\overline{\Gamma^{*}_4}$, which is shown to be a a convex cone in the entropy vector space $\mathbb{R}^{15}$. Similarly, the \emph{von Neumann cone} $\Gamma_4$ can be defined as the set of vectors that satisfy the von Neumann entropy inequalities for 4-party systems, i.e., positivity, strong subadditivity, and weak monotonicity. Because these inequalities hold for an arbitrary quantum system, it necessarily follows that $\overline{\Gamma^{*}_4}\subset\Gamma_4$. It was proven that $\overline{\Gamma^{*}_n}=\Gamma_n$ for $n\leq 3$. Therefore von Neumann inequalities completely characterize the quantum entropy cone for 3 or fewer parties.

Due to the convexity of the cones, we know that all points inside an entropy cone can be written as a linear combination of its extremal rays. The entropy cone $\Gamma_4$ produced by all von Neumann entropy inequalities is known and is characterized by the extremal rays listed in table \ref{tb:entropyTable} \cite{IbinsonThesis}.

\begin{table}
\begin{center}
\begin{tabular} {|l|l*{13}{c}r|}
\hline
~ & A & B & C & D & E & AB & AC & AD & AE & BC & BD & BE & CD & CE & DE \\ \hline
 Family 1 & 1 & 1 & 0 & 0 & 0 & 0 & 1 & 1 & 1 & 1 & 1 & 1 & 0 & 0 & 0 \\
 Family 2 & 1 & 1 & 1 & 1 & 0 & 1 & 1 & 1 & 1 & 1 & 1 & 1 & 1 & 1 & 1 \\
 Family 3 & 1 & 1 & 1 & 1 & 0 & 2 & 2 & 2 & 1 & 2 & 2 & 1 & 2 & 1 & 1 \\
 Family 4 & 1 & 1 & 1 & 1 & 1 & 1 & 1 & 1 & 1 & 1 & 1 & 1 & 1 & 1 & 1 \\
 Family 5 & 2 & 1 & 1 & 1 & 1 & 3 & 3 & 3 & 3 & 2 & 2 & 2 & 2 & 2 & 2 \\
 Family 6 & 1 & 1 & 2 & 2 & 2 & 2 & 3 & 3 & 3 & 3 & 3 & 3 & 2 & 2 & 2 \\ \hline
 Family 7 & 3 & 3 & 2 & 2 & 2 & 4 & 3 & 3 & 3 & 3 & 3 & 3 & 3 & 4 & 4 \\
 Family 8 & 3 & 3 & 3 & 3 & 2 & 4 & 4 & 4 & 5 & 4 & 4 & 5 & 6 & 5 & 5 \\
\hline
\end{tabular}

\end{center}
\caption{Families of extremal rays of the 4-party von Neumann cone constructed using known quantum inequalties. For a (mixed) state with subregions A, B, C and D, the region E is the corresponding purifying region.}
\label{tb:entropyTable}
\end{table}

It is shown in \cite{LindenWinter}, however, that families 7 and 8 are not physically constructible. Therefore,
it is suspected that $\overline{\Gamma^{*}_4}$ should be a proper subset of $\Gamma_4$ for $n\geq 4$ \cite{LindenWinter,IbinsonThesis},\footnote{In fact, studies in the classical Shannon entropy reveal that additional entropy inequalities, such as Zhang-Yeung Inequality, are needed in addition to the Shannon-type entropies.} so that additional 4-party entropy inequalities may be needed to complete the entropy cone.\footnote{There is also the possibility of a characterization in terms of non-linear inequalities. However, that is beyond the scope of this work.}
By searching through the integral linear combinations of the constrained entropy equalities constructed from graphs as described in section~\ref{sec:graph} above, we have generated a set of inequalities that satisfy families 1 through 6 but can violate families 7 and 8 for certain permutations.
After a cursory search, we found two such candidate inequalities (up to permutations):
\begin{equation}
\label{eqn:cd1}
\begin{aligned}
&S(ABD)+S(ABC)+S(CD)-S(BC)-S(AB)-S(BD)-S(AC) \\
&-S(AD)+S(B)+S(A)\leq 0
\end{aligned}
\end{equation}
\begin{equation}
\label{eqn:cd2}
\begin{aligned}
&S(ABD)+S(ABC)+S(BCD)-2S(BD)-2S(BC)+S(CD)-S(AD) \\
&-S(AC)-S(AB)+2S(B)+S(A)\leq 0
\end{aligned}
\end{equation}
Both inequalities, as well as the quantum analogue of the Zhang-Yeung inequality \cite{ZhangYeung,SixClassicalIneq},
\[ I(A:B)+I(A:CD)+3I(C:D|A)+I(C:D|B)-2I(C:D)\geq 0, \]
also satisfy the candidate extremal ray in table \ref{tb:candidateray} for the 4-party quantum entropy cone $\overline{\Gamma^{*}_4}$. We note that the graph construction, as currently formulated, cannot produce the Zhang-Yeung inequality.

\begin{table}
\begin{center}
\begin{tabular} {|l|l|l|l|l|l|l|l|l|l|l|l|l|c|r|}
\hline
 A & B & C & D & E & AB & AC & AD & AE & BC & BD & BE & CD & CE & DE \\ \hline
 1 & 1 & 2 & 2 & 2 & 2 & 2 & 2 & 2 & 2 & 2 & 2 & 2 & 2 & 2 \\ \hline
\end{tabular}
\end{center}
\caption{A candidate extremal ray for the 4-party quantum entropy cone proposed by \cite{IbinsonThesis}.}
\label{tb:candidateray}
\end{table}

Inequality~\eqref{eqn:cd1} is known as the \emph{Ingleton inequality}. It can also be written as
\[ I(A:B|C)+I(A:B|D)+I(C:D)-I(A:B)\geq 0. \]
It is known that the Ingleton inequality does not hold for general quantum states (not even for classical probability distributions), but that it is a valid inequality for the subclass of stabilizer states \cite{Stabilizer,gross2013stabilizer}.

Inequality~\eqref{eqn:cd2} on the other hand seems to be independent of the other 4-party linear candidates and to the best of our knowledge, has not been tested to a greater extent. All tests we've conducted so far on this inequality return the same result as the quantum analogue of Zhang-Yeung inequality. It will be worthwhile to generate random 5-partite quantum pure states and numerically check if the inequality can be violated. Note that such checks do not constitute a proof. However, it can be useful in finding a counterexample.

\section{Conclusion and future directions}
\label{sec:conclusion}
We here restate our findings:
\begin{enumerate}
\item We have completely characterized the entropy (in)equalities obeyed by systems in which the entanglement entropy satisfies an exact area law.
We find that such an entropy inequality is valid if and only if it is valid for the entropies of Bell pairs shared between arbitrary subsystems.
In particular, all holographic entropy inequalities, such as the cyclic inequalities established recently in \cite{BNOSSW}, are satisfied by systems with an exact area law.
These (in)equalities may provide constraining tests to determine whether certain condensed matter systems satisfy an area law.
\item The cyclic (in)equalities in two-dimensional systems with non-trivial topological order can be seen as a generalization of \cite{TEEKitaevPreskill} which extracts the topological entanglement entropy using higher number of partitions.
These higher $k$ generalizations of $S_{\text{topo}}$ are sensitive (or insensitive) to different types of deviations from area-law scaling.
\item A graph representation for constrained entropy equalities for systems with an exact area-law scaling is found. As this construction recovers a wide class of entropy equalities including strong subadditivity, it may be suspected that further quantum inequalities may also be found in the set of graph-generated equalities. Following this approach, we have found a candidate linear entropy inequality for general 4-party quantum states.
\end{enumerate}

As we have seen, the graph representation of entropies in area-law systems used in section~\ref{sec:entropy cone} offers surprisingly powerful insights.
In the absence of the minimization that appears in holography, several holographic inequalities now hold exactly as equalities for systems satisfying an exact area law, and we may understand the entropy cone spanned by the area law systems as a particular degeneration of the holographic entropy cone \cite{BNOSSW}.The method may also provide useful insight for the long-standing problem of finding linear inequalities for the entropies of general multipartite quantum states. In this regard, we also note that generalizations of the graph-theoretical approach is much desirable. One such generalization will involve constructing different graphs for a quantum state with holographic dual. We suspect that the geometry of AdS or its dual kinematic space can be effectively captured by analyzing generalized graph representations for these states. In particular, machineries developed in spectral graph drawing may be used to recover the emergent geometry for more general states.

\acknowledgments

We thank Sepehr Nezami and Bogdan Stoica for helpful discussions. This research is supported in part
by the Institute for Quantum Information and Matter at Caltech, by the Walter Burke Institute for Theoretical Physics at Caltech,
by DOE grant DE-SC0011632,
by the Gordon and Betty Moore Foundation through Grant 776 to the Caltech Moore Center for Theoretical Cosmology and Physics,
by the Simons Foundation, and by FQXI.
N.B. is supported by the DuBridge postdoctoral fellowship at the Walter Burke Institute for Theoretical Physics.

\bibliographystyle{JHEP}
\bibliography{EntropicIneq}

\end{document}